\documentclass[letterpaper, 10 pt, conference]{ieeeconf}  % Comment this line out if you need a4paper

\IEEEoverridecommandlockouts                              % This command is only needed if 
\usepackage{amsfonts}
\usepackage{amstext}
\usepackage{mathtools}
\usepackage{cases}
\usepackage{enumerate}
\usepackage{subfigure}
\usepackage{amssymb}
\usepackage{float}
\usepackage{xcolor}
\usepackage{comment}

\newtheorem{thm}{\textbf{\text{Theorem}}}
\newtheorem{lem}{\textbf{\text{Lemma}}}

\newtheorem{pb}{\textbf{\text{Problem}}}
\newtheorem{assumption}{\textbf{\text{Assumption}}}
\newtheorem{pro}{\textbf{\text{Proposition}}}

\newcommand{\ie}{\textit{i.e.}}
\newcommand{\eg}{\textit{e.g.}}

\title{\LARGE \bf
Bearing-based distributed pose estimation for multi-agent networks 
}
\author{Mouaad Boughellaba and Abdelhamid Tayebi% <-this % stops a space
\thanks{This work was supported by the National Sciences and Engineering Research Council of Canada (NSERC), under the grants NSERC-DG RGPIN 2020-06270.} 
\thanks{The authors are with the Department of Electrical Engineering, Lakehead University, Thunder Bay, ON P7B 5E1, Canada \tt\small \{mboughel,atayebi\}@lakeheadu.ca.} 
}

\begin{document}

\maketitle
\thispagestyle{empty}
\pagestyle{empty}

%%%%%%%%%%%%%%%%%%%%%%%%%%%%%%%%%%%%%%%%%%%%%%%%%%%%%%%%%%%%%%%%%%%%%%%%%%%%%%%%
\begin{abstract}
In this paper, we address the distributed pose estimation problem for multi-agent systems, where the agents have unknown static positions and time-varying orientations. The interaction graph is assumed to be directed and acyclic with two leaders that have access to their position and orientation. We propose a nonlinear distributed pose estimation scheme relying on individual angular velocity measurements and local relative (time-varying) bearing measurements. The proposed estimation scheme consists of two cascaded distributed observers, an almost globally asymptotically stable (AGAS) attitude observer and an input-to-state stable (ISS) position observer, leading to an overall AGAS distributed localization scheme. Numerical simulation results are presented to illustrate the performance of our proposed distributed pose estimation scheme.
\end{abstract}

%%%%%%%%%%%%%%%%%%%%%%%%%%%%%%%%%%%%%%%%%%%%%%%%%%%%%%%%%%%%%%%%%%%%%%%%%%%%%%%%
\section{INTRODUCTION}
The distributed pose estimation problem for multi-agent networks consists of estimating the agents' poses (positions and orientations) in a distributed manner using some available absolute and relative measurements. Due to the importance of this problem in many applications related to multi-agent autonomous networks, significant research has been devoted to designing robust and reliable distributed pose localization schemes. According to the network's sensing capabilities, these schemes can be categorized as position-based, distance-based, and bearing-based schemes. Recently, the latter category is gaining in popularity because of the revolutionary development in bearing sensors. As a result, several bearing-based distributed position estimation solutions have been proposed in the literature \cite{Zhao_auto2015,Tang_CSL2023,TANG2021109567, TANG2022110289}. However, the aforementioned references assume that the bearings are expressed in a global reference frame (\ie, knowledge of agents' orientations with respect to a global reference frame), which, unfortunately, is not the case in most of the practical applications since the bearing measurements are usually obtained locally from a sensor (\eg, a camera) mounted on the agent. This motivated many authors to design a distributed attitude observer that can be fed into a position estimation scheme together with local bearing measurements to obtain an overall cascaded bearing-based distributed pose estimation scheme \cite{li_2020, lee_2019}. The idea consists in using the estimated attitudes to transform the local relative bearing measurements into the global reference frame and then use the transformed bearings in the position estimation law. In \cite{lee_2016, lee_auto2016, Tran_CDC2018, lee_2019}, the authors proposed several orientation estimation algorithms based on the consensus approach and the Gram-Schmidt procedure. Later, these algorithms were extended to deal with a time-varying orientation in an arbitrary dimensional space \cite{Tran_TCNS2019, Tran_TCNS2020}. Note that the Gram-Schmidt procedure sometimes fails when the estimated matrix is singular. In addition, most of the distributed attitude estimation schemes require relative attitude measurements, which are difficult to obtain since no low-cost setup can provide such measurements. To the best of the authors' knowledge, there are very few results in the literature that address the distributed pose estimation problem for multi-agent networks relying on local relative bearing measurements without using the relative attitude information. For instance, the authors of \cite{TRAN2020109125,BOUGHELLABA2023110949,Mouaad_CDC2022} proposed a distributed pose estimation schemes based only on local relative bearing measurements and some absolute measurements (angular and linear velocities).\\
In this paper, motivated by \cite{TRAN2020109125,BOUGHELLABA2023110949}, we propose a distributed pose estimation scheme for multi-agent networks, relying only on the angular velocity of each agent and the local time-varying inter-agent bearing measurements available according to a directed graph topology with a leader-follower structure. A rigorous stability analysis is provided asserting that the proposed estimation scheme enjoys almost global asymptotic stability.  On top of being much simpler than the pose estimator in \cite{TRAN2020109125,BOUGHELLABA2023110949}, the proposed pose estimation schemes is endowed with almost global asymptotic stability guarantees.\\
The remainder of this paper is organized as follows: Section \ref{s2} provides some preliminaries needed in this work. In Section \ref{s3}, we formulate the distributed pose estimation problem for multi-agent networks. Section \ref{s4} presents our proposed cascaded bearing-based distributed pose estimation solution. Simulation results and concluding remarks are presented in Section \ref{s5} and Section \ref{s6}, respectively.

\section{PRELIMINARIES}\label{s2}
\subsection{Notations}
The sets of real numbers and the n-dimensional Euclidean space  are denoted by $\mathbb{R}$ and $\mathbb{R}^n$, respectively. The set of unit vectors in $\mathbb{R}^n$ is defined as $\mathbb{S}^{n-1}:=\{x\in \mathbb{R}^n~|~x^T x =1\}$. Given two matrices $A$,$B$ $\in \mathbb{R}^{m\times n}$, their Euclidean inner product is defined as $\langle \langle A,B \rangle \rangle=\text{tr}(A^T B)$. The Euclidean norm of a vector $x \in \mathbb{R}^n$ is defined as $||x||=\sqrt{x^T x}$, and
the Frobenius norm of a matrix $A \in \mathbb{R}^{n\times n}$ is given by $||A||_F=\sqrt{\langle \langle A,A \rangle \rangle}$. The matrix $I_n \in \mathbb{R}^{n \times n}$ denotes the identity matrix.\\
The attitude of a rigid body is represented by a rotation matrix $R$ which belongs to the special orthogonal group $SO(3):= \{ R\in \mathbb{R}^{3\times 3} | \hspace{0.1cm}\text{det}(R)=1, R^TR=I_3\}$. The tangent space of the compact manifold $SO(3)$ is given by $T_RSO(3):=\{R \hspace{0.1cm}\Omega \hspace{0.2cm} | \hspace{0.2cm} \Omega \in \mathfrak{so}(3)\}$, where $\mathfrak{so}(3):=\{ \Omega \in \mathbb{R}^{3\times 3} | \Omega^T=-\Omega\}$ is the Lie algebra of the matrix Lie group $SO(3)$. The map $[.]^{\times}: \mathbb{R}^3 \rightarrow \mathfrak{so}(3)$ is defined such that $[x]^\times y=x \times y$, for any $x,y \in \mathbb{R}^3$, where $\times$ denotes the vector cross product on $\mathbb{R}^3$. The inverse map of $[.]^{\times}$ is $\text{vex}: \mathfrak{so}(3) \rightarrow \mathbb{R}^3$ such that $\text{vex}([\omega]^\times)=\omega$, and $[\text{vex}(\Omega)]^\times=\Omega$ for all $\omega \in \mathbb{R}^3$ and $\Omega \in \mathfrak{so}(3)$. Let $\mathbb{P}_a : \mathbb{R}^{3\times 3} \rightarrow \mathfrak{so}(3)$ be the projection map on the Lie algebra $\mathfrak{so}(3)$ such that $\mathbb{P}_a(A):=(A-A^T)/2$. Given a 3-by-3 matrix $C:=[c_{ij}]_{i,j=1,2,3}$, one has  $\psi(C) := \text{vex} \circ \mathbb{P}_a (C)=\text{vex}(\mathbb{P}_a(C))=\frac{1}{2}[c_{32}-c_{23},c_{13}-c_{31},c_{21}-c_{12}]^T$. For any $R\in SO(3)$, the normalized Euclidean distance on $SO(3)$, with respect to the identity $I_3$, is defined as $|R|_I^2:=\frac{1}{4}\text{tr}(I_3-R)$ $\in[0,1]$. The angle-axis parameterization of $SO(3)$, is given by $\mathcal{R}_\alpha(\theta, v):=I_3+sin\hspace{0.05cm}\theta \hspace{0.2cm}[v]^\times + (1-cos\hspace{0.05cm}\theta)([v]^\times)^2$, where $v\in \mathbb{S}^2$ and  $\theta \in \mathbb{R}$ are the rotation axis and angle, respectively. The orthogonal projection map $P:\mathbb{R}^3\rightarrow\mathbb{R}^{3\times3}$ is defined as $P_x:=I_3-\frac{x x^T}{||x||^2}$, $\forall x\in\mathbb{R}^3\setminus\{0_3\}$. One can verify that $P_x^T=P_x$, $P_x^2=P_x$, and $P_x$ is positive semi-definite. Given a subset $\mathcal{S} \subset \mathbb{N}$, where $\mathbb{N}$ is a set of natural numbers, the cardinality of $\mathcal{S}$ is denoted by $|\mathcal{S}|$. 

\subsection{Graph Theory}
Consider a network of $n$ agents. The interaction topology between the agents is described by a directed graph $\mathcal G = (\mathcal V,\mathcal E)$, where $\mathcal V=\{1, \hdots, n\}$ and $\mathcal E \subseteq \mathcal V \times \mathcal V $ represent the vertex (or agent) set and the edge set, respectively. In directed graphs, $(i,j) \in \mathcal E$ does not necessarily imply $(j,i) \in \mathcal E$. The set of neighbors of agent $i$ (from which this agent receives information) is defined as $\mathcal N_i = \{j \in \mathcal V : (i,j) \in \mathcal E \}$. A directed path is a sequence of edges in a directed graph $\mathcal{G}$. A directed graph $\mathcal{G}$ is said to be acyclic if it does not have any directed path that forms a closed loop.

\section{PROBLEM STATEMENT}\label{s3}
Consider a network of $n$ agents, where the motion of each agent $i \in \mathcal{V}$ is governed by the following rotational kinematic equation:
\begin{equation}\label{pb1_rotation}
    \dot{R}_i = R_i[\omega_i]^{\times},
\end{equation}
where $R_i \in SO(3)$ represents the orientation of the body-attached frame of agent $i$ with respect to the inertial frame, and $\omega_i\in \mathbb{R}^3$ is the angular velocity of agent $i$ measured in the body-attached frame of the same agent. Let $p_i \in \mathbb{R}^3$ denote the position of agent $i$ with respect to the inertial frame. In this work, we assume that the positions of the agents are fixed and do not change with time, \ie, $\dot{p}_i=0$, for all $i \in \mathcal{V}$.\\
\noindent The measurement of the local relative bearing between agent $i$ and agent $j$ is given by
\begin{align}\label{measurement_model_1}
    b_{ij}^i&:=R_i^T b_{ij},
\end{align}
where $b_{ij}:=\frac{p_j-p_i}{||p_j-p_i||}$ and $b_{ij}^i$ are the relative bearing measurements between agent $i$ and agent $j$ expressed in the inertial frame and the body-attached frame of agent $i$, respectively. The following assumptions are needed in our design:
\begin{assumption}\label{v_ass}
   The angular velocity of each agent is available for measurement and bounded.
\end{assumption}
\begin{assumption}\label{G_ass_2l}
   By assigning a number to each agent, we assume that agents $1$ and $2$ are the leaders and the other agents are the followers. We also assume that
   \begin{enumerate}[a.]
   \item \label{a22} The directed graph $\mathcal{G}$ is acyclic and each leader has at least one directed path to each follower. 
    \item \label{a21}  The leaders know their pose, and have no neighbors \ie, $\mathcal{N}_k = \{\varnothing\}$ $\forall k=1,2$.
    \item \label{a24} Each agent $i \in \mathcal{V}_f$, where $\mathcal{V}_f:=\mathcal{V}\setminus\{1, 2\}$ denotes the set of followers, measures $b_{ij}^i$ and receives $(\hat{R}_j, \hat{p}_j, b_{ji}^j)$ from its neighbors $j\in \mathcal{N}_i$. 
    \item \label{a23} No two agents are collocated, and the set of neighbors of each agent $i \in \mathcal{V}_f$ satisfies $\mathcal{N}_i \subseteq \{1,2,3,...,i-1\}$ and $|\mathcal{N}_i|\geq 2$ with at least two non-collinear bearing vectors measured by each agent.
\end{enumerate}   
\end{assumption}

\noindent Now, we will state the problem that is considered in this work.
\begin{pb}
Consider a network of $n$ rotating agents with fixed positions. Suppose Assumptions \ref{v_ass}-\ref{G_ass_2l} are satisfied. Design a bearing-based distributed pose (position and orientation) estimation scheme endowed with almost global asymptotic stability guarantees. 
\end{pb}

\section{MAIN RESULTS}\label{s4}

\subsection{Bearing-based Distributed Attitude Estimation on $SO(3)$}
For $i\in \mathcal{V}_f$, we propose the following attitude observer on $SO(3)$:
\begin{align}\label{R_obs}
    \dot{\hat{R}}_i &= \hat{R}_i\left[\omega_i-k_R \hat{R}_i^T \sum_{j \in \mathcal{N}_i} k_{ij}(\hat{R}_j b_{ij}^j \times \hat R_i b_{ij}^i)\right]^\times,
\end{align}
where $k_R, k_{ij} >0$, $\hat{R}_i \in SO(3)$ is the estimate of $R_i$, and $\hat{R}_l=R_l$, $l\in \{1,2\}$. Under Assumption \ref{G_ass_2l}, one has $||p_i-p_j||\neq0$ and consequently the bearing measurement  $b^i_{ij}$, for every $i\in \mathcal{V}$ and $j\in \mathcal{N}_i$, is well defined. Defining the attitude estimation error $\tilde{R}_i:=R_i \hat{R}_i^T$, the last term of (\ref{R_obs}) can be rewritten as follows:
\begin{align}\label{correcting_term}
    \sum_{j \in \mathcal{N}_i} k_{ij} (\hat{R}_j b_{ij}^j \times \hat{R}_i b_{ij}^i)&=\sum_{j \in \mathcal{N}_i} k_{ij} (b_{ij} \times \tilde{R}_i^T b_{ij})\nonumber\\
    +&\sum_{j \in \mathcal{N}_i} k_{ij}\left((\tilde{R}_j^T-I_3) b_{ij} \times \tilde{R}_i^T b_{ij}\right)\nonumber\\
    =&-2\psi(M_i\tilde{R}_i)\nonumber\\
    +&\sum_{j \in \mathcal{N}_i} k_{ij}\left((\tilde{R}_j^T-I_3) b_{ij} \times \tilde{R}_i^T b_{ij}\right),
\end{align}
where $M_i:=\sum_{j \in \mathcal{N}_i} k_{ij} b_{ij} b_{ij}^T$. The last equation was obtained using the fact $x \times y = 2 \psi(y x^T)$, $\forall x, y \in \mathbb{R}^3$. It is always possible to choose $k_{ij}>0$ such that the matrix $M_i$ is positive semi-definite with three distinct eigenvalues. From \eqref{pb1_rotation} and \eqref{R_obs}, it follows that the time derivative of the attitude estimation error, for every $i\in \mathcal{V}_f$, is given by
 \begin{align}%\label{error_rotation_dynamics_11}
    \dot{\tilde{R}}_i =& -R_i\left[\omega_i-k_R \hat{R}_i^T \sum_{j \in \mathcal{N}_i} k_{ij}(\hat{R}_j b_{ij}^j \times \hat R_i b_{ij}^i)\right]^\times\hat{R}_i^T\nonumber\\
    &+R_i[\omega_i]^\times \hat{R}_i^T.
\end{align}
Since $[x+y]^\times=[x]^\times+[y]^\times$ and $[Rx]^\times = R[x]^\times R^T$, $\forall x, y\in \mathbb{R}^3$ and $R\in SO(3)$, and in view of \eqref{correcting_term}, one can simplify the last equation as follows:
\begin{equation}\label{forced_sys}
    \dot{\tilde R}_i = -2k_R \tilde R_i\left[\psi(M_i \tilde R_i)\right]^\times+k_R\tilde R_i\left[\sum_{j \in \mathcal{N}_i} k_{ij}g_{ij}(\tilde R_j)\right]^\times,
\end{equation}
where $g_{ij}(\tilde R_j):=(\tilde R^T_j-I_3)b_{ij} \times \tilde R_i^T b_{ij}$. Note that the term $g_{ij}(\tilde R_j)$ is bounded and vanishes when $\tilde R_j = I_3$. Furthermore, system \eqref{forced_sys} can be viewed as a cascaded system, where the attitude estimation errors of the neighbors are considered as inputs to the following unforced system:
\begin{equation}\label{unforced_sys}
    \dot{\tilde R}_i = -2k_R \tilde R_i\left[\psi(M_i \tilde R_i)\right]^\times.
\end{equation}
The following lemma provides the stability properties of the equilibria of system \eqref{unforced_sys}. 
\begin{lem}\label{lemma1}
   Let $k_{ij}>0$ such that $M_i$ is positive semi-definite with three distinct eigenvalues. Then, the following statements hold for all $i\in \mathcal{V}_f$:
   \begin{enumerate}[i)]
       \item All solutions of \eqref{unforced_sys} converge to the following set of isolated equilibria: $\Upsilon: =\{I_3\}\cup\{\tilde R_i=\mathcal{R}_\alpha (\pi ,v_i) | v_i \in \mathcal{E}(M_i)\}\nonumber$, where $\mathcal{E}(M_i)\subset \mathbb{S}^2$ is the set of unit eigenvectors of matrix $M_i$. \label{set_of_equil}
       \item The desired equilibrium $\tilde R_i=I_3$ is locally exponentially stable. \label{des_equi}
       \item The linearized system of (\ref{unforced_sys}), at each undesired equilibrium $\Upsilon/\{I_3\}$, has at least one positive eigenvalue. \label{undes_equi} 
       \item The undesired equilibria $\Upsilon/\{I_3\}$ are unstable and the desired equilibrium $\tilde R_i=I_3$ is AGAS. \label{stability_of_equilibrium}
   \end{enumerate}
\end{lem}
\begin{proof}
 Consider the following Lyapunov function candidate:
   \begin{equation}
       \mathcal L_i = \frac{1}{4}\text{tr}\left(M_i\left(I_3-\tilde R_i\right)\right),
   \end{equation}
   whose time-derivative, along the trajectories of \eqref{unforced_sys}, is given by
   \begin{align}
       \dot{\mathcal{L}}_i = \frac{k_R}{2} \text{tr}\left(M_i \tilde R_i \left[\psi(M_i\tilde{R}_i)\right]^\times \right).
   \end{align}
   Using the facts that $\text{tr}\left(B[x]^\times\right)=\text{tr}\left(\mathbb{P}_a(B)[x]^\times\right)$ and $\text{tr}\left([x]^\times [y]^\times\right)=-2x^Ty$, for every $x, y\in \mathbb{R}^3$ and $B \in \mathbb{R}^{3\times3}$, one has
   \begin{equation}
       \dot{\mathcal{L}}_i= - k_R ||\psi(M_i \tilde R_i)||^2 \leq 0.
   \end{equation}
   %Thus, the desired equilibrium $\tilde R_i = I_3$ of \eqref{unforced_sys} is stable. 
   Since system \eqref{unforced_sys} is autonomous, by virtue of LaSalle's invariance theorem, the attitude error $\tilde R_i$ should converge to the largest invariant set contained in the set characterized by $\dot{\mathcal{L}}_i= 0$, \ie, $\psi(M_i \tilde R_i)=0$. As per \cite[Lemma 2]{Mayhew_ACC2011}, $\psi(M_i \tilde R_i)=0$ implies that $\tilde R_i \in \Upsilon$.\\
   Since the matrix $M_i$ is positive semi-definite with three distinct eigenvalues, it follows that the equilibrium points, in the set $\Upsilon$, are isolated. Moreover, following the same arguments as in \cite[Theorem 1]{Miaomiao_TAC2021}, one can show that the desired equilibrium $\tilde R_i = I_3$ is locally exponentially stable and the dynamics of the first order approximation of $\tilde R_i$ around each undesired equilibrium has at least one positive eigenvalue. Accordingly, the desired equilibrium $\tilde R_i=I_3$ of \eqref{unforced_sys} is AGAS. This completes the proof.
\end{proof}
In the next lemma, we will study the ISS property of the forced attitude error dynamics \eqref{forced_sys}, with respect to its inputs (\ie, $\tilde R_j$ with $j\in \mathcal{N}_i$), using the notion of almost global ISS introduced in \cite{Angeli_TAC2011}.  
\begin{lem}\label{lemma2}
   Let $M_i$ be positive semi-definite with three distinct eigenvalues. Then, system (\ref{forced_sys}) is almost globally ISS, for each $i\in \mathcal{V}_f$, with respect to $I_3$ and inputs $\tilde R_j$.
\end{lem}

\begin{proof}
Since system (\ref{forced_sys}), subject to the bounded input $\sum_{j \in \mathcal{N}_i} k_{ij}g_{ij}(\tilde R_j)$, evolves on the compact manifold  $SO(3)$, condition A0, given in \cite{Angeli_TAC2011}, is fulfilled. Moreover,  according to Lemma \ref{lemma1}, conditions A1 and A2, given in \cite{Angeli_TAC2011}, are also fulfilled.\\
Now, consider the following real-valued function:
\begin{align}\label{lk}
    |\tilde{R}_i|_I^2 = \frac{1}{4}\text{tr}(I_3-\tilde{R}_i),
\end{align}
\noindent whose time-derivative, along the trajectories of \eqref{forced_sys}, is given by
\begin{equation}
    \frac{d}{dt}|\tilde R_i|_I^2=-k_R \psi(\tilde R_i)^T \left( \psi(M_i \tilde R_i)-\frac{1}{2}\sum_{j \in \mathcal{N}_i} k_{ij}g_{ij}(\tilde R_j)\right).
\end{equation}
The last equation was obtained using the following identities: $\text{tr}\left(B [x]^\times\right)=\text{tr}\left(\mathbb{P}_a(B) [x]^\times\right)$ and $\text{tr}\left([x]^\times [y]^\times\right)=-2x^Ty$, for every $x, y\in \mathbb{R}^3$ and $B \in \mathbb{R}^{3\times3}$. Moreover, since $\psi(\tilde R_i)^T \psi(M_i\tilde R_i) = \psi(\tilde R_i)^T Q_i \psi(\tilde R_i)$, with $Q_i:= \sum_{j \in \mathcal{N}_i} k_{ij}\left([b_{ij}]^\times\right)^T [b_{ij}]^\times$, one has
\begin{equation}
    \frac{d}{dt}|\tilde R_i|_I^2=-k_R \psi(\tilde R_i)^T \left(Q_i \psi(\tilde R_i)-\frac{1}{2}\sum_{j \in \mathcal{N}_i} k_{ij}g_{ij}(\tilde R_j)\right).
\end{equation}
Using the fact that $||\psi(\tilde R_i)||^2=4(1-|\tilde R_i|^2_I)|\tilde R_i|^2_I \leq 1$ and at least two bearing vectors are noncollinear, one obtains
\begin{align}\label{ISS}
    \frac{d}{dt}|\tilde R_i|_I^2&\leq -4 k_R \underline{\lambda}^{Q_i}(1-|\tilde R_i|^2_I)|\tilde R_i|^2_I\nonumber\\
    &+\frac{k_R}{2}||\sum_{j \in \mathcal{N}_i} k_{ij}g_{ij}(\tilde R_j)||\nonumber\\
    &\leq  -4 k_R\underline{\lambda}^{Q_i} |\tilde R_i|^2_I+4 k_R \underline{\lambda}^{Q_i}\nonumber\\
    &+\frac{k_R}{2}\sum_{j \in \mathcal{N}_i} k_{ij}||g_{ij}(\tilde R_j)||,
\end{align}
where $\underline{\lambda}^{Q_i}$ is the smallest eigenvalue of $Q_i$, which is positive definite under the assumption that at least two bearing vectors are noncollinear. Furthermore, using the fact that $2\sqrt{2}|R|_I=||I-R||_F$, one verifies that $||g_{ij}(\tilde R_j)|| \leq 2\sqrt{2}|\tilde R_j|_I$. Hence, it follows from \eqref{ISS} that
\begin{equation}\label{ISS2}
    \frac{d}{dt}|\tilde R_i|_I^2\leq -4 k_R\underline{\lambda}^{Q_i} |\tilde R_i|^2_I+4 k_R \underline{\lambda}^{Q_i}+\sqrt{2}k_R\sum_{j \in \mathcal{N}_i} k_{ij}|\tilde R_j|_I.
\end{equation}
It is clear that, in view of (\ref{ISS2}), system (\ref{forced_sys}) satisfies the ultimate
boundedness property defined in \cite[Proposition 3]{Angeli_TAC2011}. Consequently, as per \cite[Proposition 2]{Angeli_TAC2011}, one can conclude that system (\ref{forced_sys}) is almost globally ISS with respect to $I_3$ and inputs $\tilde R_j$.\\
\end{proof}
In the remainder of this section, we will prove almost global asymptotic stability of the equilibrium  $(\tilde R_3=I_3, \tilde R_4=I_3, \hdots, \tilde R_n=I_3)$ of the $n$-agent network governed by the cascaded dynamics \eqref{forced_sys}. Thanks to the cascaded structure of the interaction graph topology $\mathcal{G}$, as per Assumption \ref{G_ass_2l}, a mathematical induction argument, together with Lemma \ref{lemma1} and Lemma \ref{lemma2}, can be used to prove the result.\\

\noindent\textbf{\textit{Stability analysis of a network with one follower}}\\
For a network with one follower, according to \eqref{forced_sys}, one has the following attitude dynamics:
\begin{align} \label{R_3}
    \dot{\tilde{R}}_3 &= -2k_R \tilde R_3\left[\psi(M_3 \tilde R_3)\right]^\times.
\end{align}
where $M_3$ is positive semi-definite with three distinct eigenvalues. Note that $\sum_{j \in \mathcal{N}_3} k_{3j}g_{3j}(\tilde R_j)=0$ since $\mathcal{N}_3=\{1, 2\}$ and $\hat R_l=R_l$ (\ie, $\tilde R_l=I
_3$), $l\in\{1, 2\}$, as par Assumption \ref{G_ass_2l}.  It follows from Lemma \ref{lemma1} that the equilibrium point $\tilde R_3=I_3$ of system \eqref{R_3} is AGAS.\\

\noindent\textbf{\textit{Stability analysis of a network with two followers}}\\
The attitude dynamics of a network with two followers are given by
\begin{align} 
    \dot{\tilde{R}}_3 &= -2k_R \tilde R_3\left[\psi(M_3 \tilde R_3)\right]^\times \label{R_34}\\
    \dot{\tilde R}_4 &= -2k_R \tilde R_4\left[\psi(M_4 \tilde R_4)\right]^\times+k_R\tilde R_4\left[\sum_{j \in \mathcal{N}_4} k_{4j}g_{4j}(\tilde R_j)\right]^\times,\label{R_44}
\end{align}
where $M_3$ and $M_4$ are positive semi-definite with three distinct eigenvalues.
\begin{pro}\label{pro1}
 Under Assumption \ref{G_ass_2l}, the equilibrium point $(\tilde R_3=I_3, \tilde R_4=I_3)$ of system \eqref{R_34}-\eqref{R_44} is AGAS.  
\end{pro}
\begin{proof}
    According to Assumption \ref{G_ass_2l}, one has either $3 \in \mathcal{N}_4$ or $3 \notin \mathcal{N}_4$. In the case where $3 \notin \mathcal{N}_4$, the $\tilde R_3$-subsystem and the $\tilde R_4$-subsystem are independent, and the almost global asymptotic stability of the equilibrium point $(\tilde R_3=I_3, \tilde R_4=I_3)$ can be directly deduced from Lemma \ref{lemma1}. On the other hand, if $3 \in \mathcal{N}_4$, the two subsystems ($\tilde R_3$-subsystem and $\tilde R_4$-subsystem) are cascaded, and as such, in view of Lemma \ref{lemma1} and Lemma \ref{lemma2}, it follows that the $\tilde R_4$-subsystem is almost globally ISS with respect to $I_3$ and input $\tilde R_3$, and the $\tilde R_3$-subsystem is AGAS at $\tilde R_3=I_3$. Finally, in view of \cite[Theorem 2]{Angeli_TAC2004}, it follows that the cascaded system \eqref{R_34}-\eqref{R_44} is AGAS at $(\tilde R_3=I_3, \tilde R_4=I_3)$.\\  
\end{proof}

\noindent\textbf{\textit{Stability analysis of a network with $n-2$~followers}}\\
Consider the following attitude dynamics of a network with $n-2$~followers:
{\small
\begin{align} 
    &\dot{\tilde{R}}_3 = -2k_R \tilde R_3\left[\psi(M_3 \tilde R_3)\right]^\times \label{R_3n}\\
    &\dot{\tilde R}_4 = -2k_R \tilde R_4\left[\psi(M_4 \tilde R_4)\right]^\times+k_R\tilde R_4\left[\sum_{j \in \mathcal{N}_4} k_{4j}g_{4j}(\tilde R_j)\right]^\times\label{R_4n}\\
    &~~~~\vdots\nonumber\\
    &\dot{\tilde R}_{n-1} = -2k_R \tilde R_{n-1}\left[\psi(M_{n-1} \tilde R_{n-1})\right]^\times\nonumber\\
    &~~~~~~~~~~~~~+k_R\tilde R_{n-1}\left[\sum_{j \in \mathcal{N}_{n-1}} k_{(n-1)j}g_{(n-1)j}(\tilde R_j)\right]^\times\label{R_n1n}\\
    &\dot{\tilde R}_n = -2k_R \tilde R_n\left[\psi(M_n \tilde R_n)\right]^\times+k_R\tilde R_n\left[\sum_{j \in \mathcal{N}_n} k_{nj}g_{nj}(\tilde R_j)\right]^\times, \label{R_nn}
\end{align}}where $M_i$, $i \in \mathcal{V}_f$, is positive semi-definite with three distinct eigenvalues. Now, we can formally state the stability properties of the equilibrium point $(\tilde R_3=I_3, \tilde R_4=I_3, \hdots, \tilde R_n=I_3)$ of the $n$-agent cascaded system \eqref{R_3n}-\eqref{R_nn}.
\begin{thm}\label{thm_1}
    Under Assumption \ref{G_ass_2l}, the equilibrium point $(\tilde R_3=I_3, \tilde R_4=I_3, \hdots, \tilde R_n=I_3)$ of the $n$-agent cascaded system \eqref{R_3n}-\eqref{R_nn} is AGAS. 
\end{thm}
\begin{proof}
    We will prove the claimed result by induction. First, it follows from Proposition \ref{pro1} that the cascaded system \eqref{R_3n}-\eqref{R_4n} is AGAS at $(\tilde R_3=I_3, \tilde R_4=I_3)$. Second, we assume that the cascaded $(n-1)$-agent subsystem \eqref{R_3n}-\eqref{R_n1n} is AGAS at $(\tilde R_3=I_3, \tilde R_4=I_3, \hdots, \tilde R_{n-1}=I_3)$. Finally, using the proof by induction, with the fact that the $\tilde R_n$-subsystem is almost globally ISS with respect to $I_3$ and inputs from the cascaded $(n-1)$-agent subsystem \eqref{R_3n}-\eqref{R_n1n}, one can show that the $n$-agent cascaded system is AGAS at $(\tilde R_3=I_3, \tilde R_4=I_3, \hdots, \tilde R_n=I_3)$ according to \cite{Angeli_TAC2004}. This completes the proof.       
\end{proof}
\subsection{Bearing-based Distributed Pose Estimation}
Consider the distributed attitude observer \eqref{R_obs} together with the following distributed position estimation law:
\begin{align}\label{p_obs}
        \Dot{\hat{p}}_i =&-k_R \Bigg[\sum_{j \in \mathcal{N}_i} k_{ij}  \left(\hat{R}_j b_{ij}^j \times \hat{R}_i b_{ij}^i\right)\Bigg]^\times \hat{p}_i\nonumber\\
        &- k_p \sum_{j\in \mathcal{N}_i}  \hat{R}_i P_{b_{ij}^i} \hat{R}_i^T(\hat{p}_i-\hat{p}_j),
\end{align}
for each $i\in \mathcal{V}_f$, where $k_p, k_R, k_{ij} > 0$, $\hat{p}_i \in \mathbb{R}^3$ is the estimate of $p_i$, $\hat{R}_i \in SO(3)$ is the estimate of $R_i$ obtained from the distributed attitude observer \eqref{R_obs}, and $(\hat{R}_l, \hat{p}_l) = (R_l, p_l)$, $l\in \{1,2\}$. Define the position estimation error as $\tilde{p}_i:=p_i - \tilde{R}_i \hat{p}_i$. Its time derivative, in view of (\ref{forced_sys}) and (\ref{p_obs}), is given by
\begin{align}\label{p_1}
    \dot{\tilde{p}}_i =-k_p\sum_{j\in \mathcal{N}_i}P_{b_{ij}} \tilde{p}_i +k_p\sum_{j\in \mathcal{N}_i}P_{b_{ij}}(p_j-\tilde{R}_i \hat{p}_j),
\end{align}
with $i \in \mathcal{V}_f$. We have used the fact that $P_{b_{ij}}(p_i-p_j)=0$ and $P_{b^i_{ij}}=R_i^T P_{b_{ij}}R_i$ to obtain the last equality. It follows from \eqref{p_1} that
\begin{equation}\label{p_forced_sys}
    \dot{\tilde{p}}_i =-k_p\sum_{j\in \mathcal{N}_i}P_{b_{ij}} \tilde{p}_i +k_p\sum_{j\in \mathcal{N}_i}P_{b_{ij}}f_j(\tilde p_j, \tilde R_j, \tilde R_i),
\end{equation}
where $f_j(\tilde p_j, \tilde R_j, \tilde R_i):=\left((\tilde R_j-I_3)-(\tilde R_i-I_3)\right)\tilde R_j^T(p_j-\tilde p_j)+\tilde p_j$. It is clear that $f_i(\tilde p_j, \tilde R_j, \tilde R_i)=0$ for $\tilde p_j=0$ and $\tilde R_j=\tilde R_i=I_3$. Again, system \eqref{p_forced_sys} can be seen as a cascaded system, where the attitude and position estimation errors of the neighbors as well as the attitude estimation error of agent $i$ are considered as inputs to the following unforced system:
\begin{equation}\label{p_unforced_sys}
    \dot{\tilde{p}}_i =-k_p\sum_{j\in \mathcal{N}_i}P_{b_{ij}} \tilde{p}_i.
\end{equation}
Next, we study the stability of the equilibrium point $\tilde p_i=0$ of system \eqref{p_unforced_sys} and the ISS property of system \eqref{p_forced_sys}.
\begin{pro}\label{pro2}
    Consider system \eqref{p_unforced_sys} under Assumption \ref{G_ass_2l}. The equilibrium point $\tilde p_i =0$, $i \in \mathcal{V}_f$, is globally exponentially stable (GES).
\end{pro}
\begin{proof}
 Under the assumption that at least two bearing vectors are non-collinear (Assumption \ref{G_ass_2l}), for every $i\in \mathcal{V}_f$, the matrix $\sum_{j\in \mathcal{N}_i}P_{b_{ij}}$ is positive definite, and hence, the equilibrium $\tilde{p}_i = 0$ of the unforced position error dynamics \eqref{p_unforced_sys} is GES.\\
\end{proof}
\begin{lem}\label{lemma3}
   Suppose Assumption \ref{G_ass_2l} is satisfied. Then, for every $i\in \mathcal{V}_f$,  system (\ref{p_forced_sys}) is ISS with respect to its inputs $\tilde p_j$, $\tilde R_j$ and $\tilde R_i$.
\end{lem}
\begin{proof}
Consider the following Lyapunov function candidate:
\begin{align}\label{vpk}
    V_i = \frac{1}{2} \tilde{p}_i^T \tilde{p}_i,
\end{align}
whose time-derivative, along the trajectories of (\ref{p_forced_sys}), is given by
\begin{align}
    \dot{V}_i=\tilde p_i^T\left(-k_p\sum_{j\in \mathcal{N}_i}P_{b_{ij}} \tilde{p}_i +k_p\sum_{j\in \mathcal{N}_i}P_{b_{ij}}f_j(\tilde p_j, \tilde R_j, \tilde R_i)\right).\nonumber
\end{align}
Since $||Ax||\leq ||A||_F ||x||$, for every $A \in \mathbb{R}^{3\times3}$ and $x \in \mathbb{R}^3$, and in view of the positive definitness of the matrix $\sum_{j\in \mathcal{N}_i}P_{b_{ij}}$ (implied from Assumption \ref{G_ass_2l}), one has
{\small
\begin{align}\label{v_i}
    \dot{V}_i \leq -k_p \underline{\lambda}_i^{P}||\tilde p_i||^2+k_p \Bar{P} \sum_{j\in \mathcal{N}_i} ||\tilde p_i||~||f_j(\tilde p_j, \tilde R_j, \tilde R_i)||,
\end{align}
}
where $\underline{\lambda}_i^{P}$ denotes the smallest eigenvalue of the matrix $\sum_{j\in \mathcal{N}_i}P_{b_{ij}}$ and $\Bar{P}$ denotes the upper bound of the projection matrix norm, \ie, $||P_{b_{ij}}||_F\leq \Bar{P}$. Applying Young’s inequality on the last two terms of \eqref{v_i}, leads to
{\small
\begin{align}
    \dot{V}_i \leq &-k_p \underline{\lambda}_i^{P}||\tilde p_i||^2+k_p \Bar{P} \sum_{j\in \mathcal{N}_i}\Big(\xi_i||\tilde p_i||^2\nonumber\\
    &\hspace{3cm}+\frac{1}{4 \xi_i}||f_j(\tilde p_j, \tilde R_j, \tilde R_i)||^2\Big)\\
    \leq &-k_p \big(\underline{\lambda}_i^{P}-\Bar{P}\xi_i|\mathcal{N}_i|\big)||\tilde p_i||^2+\frac{k_p \Bar{P} }{4 \xi_i}\sum_{j\in \mathcal{N}_i}||f_j(\tilde p_j, \tilde R_j, \tilde R_i)||^2. 
\end{align}}Choosing $0<\xi_i<\frac{\underline{\lambda}_i^{P}}{\Bar{P}|\mathcal{N}_i|}$, for every $i \in \mathcal{V}_f$, and using the fact that $8|R|_I^2=||I-R||^2_F$ and $|R|_I\leq 1$, for every $R\in SO(3)$, one can show that 
{\small
\begin{align}\label{p_ISS}
    \dot{V}_i \leq &-\alpha_1(||\tilde p_i||)+\sum_{j\in \mathcal{N}_i}\Big(\alpha_2(||\tilde p_j||)+\alpha_3(|\tilde R_j|_I)+\alpha_4(|\tilde R_i|_I)\Big), 
\end{align}}where $\alpha_k(.)\in \mathcal{K}_\infty$ with $k\in \{1, 2, 3, 4\}$. It follows from \eqref{p_ISS} that system \eqref{p_forced_sys} is ISS with respect to inputs $\tilde p_j$, $\tilde R_j$ and $\tilde R_i$.\\
\end{proof}
 
\noindent Thanks again to the cascaded nature of the interaction graph topology $\mathcal{G}$ which allows the use of a similar induction proof as in the previous section, together with Proposition \ref{pro2} and Lemma \ref{lemma3} as well as a result from Theorem \ref{thm_1}, to establish the stability property of the equilibrium point $(\tilde R_3=I_3, \tilde R_4=I_3, \hdots, \tilde R_n=I_3, \tilde p_3=0, \tilde p_4=0, \hdots, \tilde p_n=0)$ of the $n$-agent network governed by the cascaded dynamics \eqref{forced_sys} and \eqref{p_forced_sys}.\\

\noindent\textbf{\textit{Stability analysis of a network with one follower}}\\
Consider the attitude error dynamics \eqref{R_3} cascaded with the following position error dynamics:
\begin{equation}\label{p_3}
    \dot{\tilde{p}}_3 =-k_p\sum_{j\in \mathcal{N}_3}P_{b_{3j}} \tilde{p}_3 +k_p\sum_{j\in \mathcal{N}_3}P_{b_{3j}}f_j(\tilde p_j, \tilde R_j, \tilde R_3),
\end{equation}
where $f_j(\tilde p_j, \tilde R_j, \tilde R_3)=-(\tilde R_3-I_3) p_j$ since $(\hat{R}_l, \hat{p}_l) = (R_l, p_l)$, $l\in \{1,2\}$, as per Assumption \ref{G_ass_2l}.
\begin{pro}\label{pro3}
    Suppose Assumption \ref{G_ass_2l} is satisfied. Then, the cascaded system \eqref{R_3} and \eqref{p_3} is AGAS at $(\tilde R_3=I_3, \tilde p_3=0)$.
\end{pro}
\begin{proof}
    From Lemma \ref{lemma1} and Proposition \ref{pro2}, one can conclude that the $\tilde R_3$-subsystem is AGAS at $\tilde R_3=I_3$ and the $\tilde p_3$-subsystem, with $\tilde R_3=I_3$, is GES at $\tilde p_3=0$. Thus, as per Lemma \ref{lemma3}, the cascaded system \eqref{R_3} and \eqref{p_3} is AGAS at $(\tilde R_3=I_3, \tilde p_3=0)$.\\   
\end{proof}
\noindent\textbf{\textit{Stability analysis of a network with two followers}}\\
Consider the attitude error dynamics \eqref{R_34}-\eqref{R_44} cascaded with the following position error dynamics:
\begin{align}
    \dot{\tilde{p}}_3 &=-k_p\sum_{j\in \mathcal{N}_3}P_{b_{3j}} \tilde{p}_3 +k_p\sum_{j\in \mathcal{N}_3}P_{b_{3j}}f_j(\tilde p_j, \tilde R_j, \tilde R_3) \label{p_34}\\
    \dot{\tilde{p}}_4 &=-k_p\sum_{j\in \mathcal{N}_4}P_{b_{4j}} \tilde{p}_4 +k_p\sum_{j\in \mathcal{N}_4}P_{b_{4j}}f_j(\tilde p_j, \tilde R_j, \tilde R_4). \label{p_44}
\end{align}
\begin{pro}\label{pro4}
    Suppose Assumption \ref{G_ass_2l} is satisfied. Then, the cascaded system \eqref{R_34}-\eqref{R_44} and \eqref{p_34}-\eqref{p_44} is AGAS at $(\tilde R_3=I_3, \tilde R_4=I_3, \tilde p_3=0, \tilde p_4=0)$.
\end{pro}
\begin{proof}
    It follows from Propositions \ref{pro1} and \ref{pro3} that the equilibrium point $(\tilde R_3 =I_3, \tilde R_4=I_3)$ of \eqref{R_34}-\eqref{R_44} is AGAS, and the equilibrium point $(\tilde R_3 =I_3, \tilde p_3=0)$ of \eqref{R_34} and \eqref{p_34} is AGAS. Therefore, using the fact that the $\tilde p_4$-subsystem is ISS with respect to $\tilde p_3$, $\tilde R_3$, and $\tilde R_4$, one can establish the claim in Proposition \ref{pro4}.\\ 
\end{proof}

\noindent\textbf{\textit{Stability analysis of a network with $n-2$~followers}}\\
For a network with $n-2$~followers, we consider the attitude error dynamics \eqref{R_3n}-\eqref{R_nn} cascaded with the following position error dynamics:
\begin{align}
    &\dot{\tilde{p}}_3 =-k_p\sum_{j\in \mathcal{N}_3}P_{b_{3j}} \tilde{p}_3 +k_p\sum_{j\in \mathcal{N}_3}P_{b_{3j}}f_j(\tilde p_j, \tilde R_j, \tilde R_3) \label{p_3n}\\
    &\dot{\tilde{p}}_4 =-k_p\sum_{j\in \mathcal{N}_4}P_{b_{4j}} \tilde{p}_4 +k_p\sum_{j\in \mathcal{N}_4}P_{b_{4j}}f_j(\tilde p_j, \tilde R_j, \tilde R_4) \label{p_4n}\\ 
    &~~~~\vdots\nonumber\\
    &\dot{\tilde{p}}_{n-1} =-k_p\sum_{j\in \mathcal{N}_{n-1}}P_{b_{(n-1)j}} \tilde{p}_{n-1} \nonumber\\
    &~~~~~~~~~~~+k_p\sum_{j\in \mathcal{N}_{n-1}}P_{b_{(n-1)j}}f_j(\tilde p_j, \tilde R_j, \tilde R_{n-1}) \label{p_n1n}\\ 
    &\dot{\tilde{p}}_n =-k_p\sum_{j\in \mathcal{N}_n}P_{b_{nj}} \tilde{p}_n +k_p\sum_{j\in \mathcal{N}_n}P_{b_{nj}}f_j(\tilde p_j, \tilde R_j, \tilde R_n). \label{p_nn}
\end{align}
In the following theorem, we establish the stability properties of the overall rotational and transnational cascaded estimation scheme.
\begin{thm}\label{thm_2}
    Considering the cascaded attitude and position estimation schemes given by \eqref{R_obs} and \eqref{p_obs}, respectively, where Assumption \ref{G_ass_2l} is satisfied. Suppose that the result in Theorem \ref{thm_1} holds. Then, the equilibrium point $(\tilde R_3=I_3, \tilde R_4=I_3, \hdots, \tilde R_n=I_3, \tilde p_3=0, \tilde p_4=0, \hdots, \tilde p_n=0)$ of the overall cascaded system \eqref{R_3n}-\eqref{R_nn} and \eqref{p_3n}-\eqref{p_nn} is AGAS.
\end{thm}
\begin{proof}
    Similar to the proof of Theorem \ref{thm_1}, we will prove the result of this theorem by induction. First, from Proposition \ref{pro3}, it is clear that the equilibrium point $(\tilde R_3=I_3, \tilde p_3=0)$ of \eqref{R_3n} and \eqref{p_3n} is AGAS. Second, we assume that the cascaded $(n-1)$-agent subsystem \eqref{R_3n}-\eqref{R_n1n} and \eqref{p_3n}-\eqref{p_n1n} is AGAS at $(\tilde R_3=I_3, \tilde R_4=I_3, \hdots, \tilde R_{n-1}=I_3, \tilde p_3=0, \tilde p_4=0, \hdots, \tilde p_{n-1}=0)$. Finally, using the induction arguments, in view of the result from Theorem \ref{thm_1} and the fact that the $\tilde p_n$-subsystem is ISS with respect to $\tilde R_n$ as well as inputs from the cascaded $(n-1)$-agent subsystem \eqref{R_3n}-\eqref{R_n1n} and \eqref{p_3n}-\eqref{p_n1n}, one can show that the equilibrium point $(\tilde R_3=I_3, \tilde R_4=I_3, \hdots, \tilde R_n=I_3, \tilde p_3=0, \tilde p_4=0, \hdots, \tilde p_n=0)$ of the overall cascaded system \eqref{R_3n}-\eqref{R_nn} and \eqref{p_3n}-\eqref{p_nn} is AGAS. This completes the proof.            
\end{proof}

%\vspace{-1.9cm}
\section{SIMULATION}\label{s5}
In this section, we present some numerical simulations to illustrate the performance of our proposed distributed pose estimation scheme. We consider an eight-agent system in a $3$-dimensional space with the following positions: $p_1=[0~0~0]^T$, $p_2=[2~0~0]^T$, $p_3=[2~2~0]^T$, $p_4=[0~2~0]^T$, $p_5=[0~0~2]^T$, $p_6=[2~0~2]^T$, $p_7=[2~2~2]^T$ and $p_8=[0~2~2]^T$. The rotational subsystem is driven by the following angular velocities: $\omega_1=[1 -2~1]^T$, $\omega_2(t)=[-\cos{3t}~1~\sin{2t}]^T$, $\omega_3(t)=[-\cos{t}~1~\sin{2t}]^T$, $\omega_4(t)=[-\cos{2t}~1~\sin{5t}]^T$, $\omega_5(t)=[-\cos{5t}~1~\sin{9t}]^T$, $\omega_6(t)=[-\cos{2t}~\sin{9t}~1]^T$, $\omega_7(t)=[-\cos{4t}~1~2]^T$ and $\omega_8(t)=[-2~1~\sin{9t}]^T$. The initial rotations of all agents are chosen to be the identity.\\
We use a directed graph with a leader-follower structure (see Assumption \ref{G_ass_2l}) to model the interaction graph of the eight-agent system as it is shown in Figure \ref{graph}. Accordingly, the neighbors sets of the agents are given as $\mathcal{N}_1=\mathcal{N}_2 = \{\varnothing\}$, $\mathcal{N}_3 = \{1, 2\}$, $\mathcal{N}_4 = \{2, 3\}$, $\mathcal{N}_5 = \{1, 4\}$, $\mathcal{N}_6 = \{2, 4, 5\}$, $\mathcal{N}_7 = \{3, 4, 6\}$ and $\mathcal{N}_8 = \{1, 7\}$. The initial conditions of our proposed estimation scheme are chosen as: $\hat{p}_3(0)=[-2~0~-1]^T$, $\hat{p}_4(0)=[-1~2~2]^T$, $\hat{p}_5(0)=[-2~2~4]^T$, $\hat{p}_6(0)=[0~0~0]^T$, $\hat{p}_7(0)=[-4~0~1]^T$, $\hat{p}_8(0)=[-3~\frac{1}{2}~2]^T$, $\hat{R}_3(0)=\mathcal{R}_\alpha(0.1\pi,v)$, $\hat{R}_4(0)=\mathcal{R}_\alpha(0.2\pi,v)$, $\hat{R}_5(0)=\mathcal{R}_\alpha(0.3\pi,v)$, $\hat{R}_6(0)=\mathcal{R}_\alpha(0.9\pi,v)$, $\hat{R}_7(0)=\mathcal{R}_\alpha(0.4\pi,v)$ and $\hat{R}_8(0)=\mathcal{R}_\alpha(0.5\pi,v)$ with $v=[1~0~0]^T$. The gain parameters are taken as follows: $k_p=1$, $k_R= 1$ and $k_{ij}= 1$ for every $(i, j)\in \mathcal{E}$. The time evolution of the average attitude and position estimation error norms are provided in Figure \ref{rotation_leader_follower} and Figure \ref{translation_leader_follower}, respectively.
\begin{figure}[h]
    \centering
\includegraphics[width=0.6\linewidth]{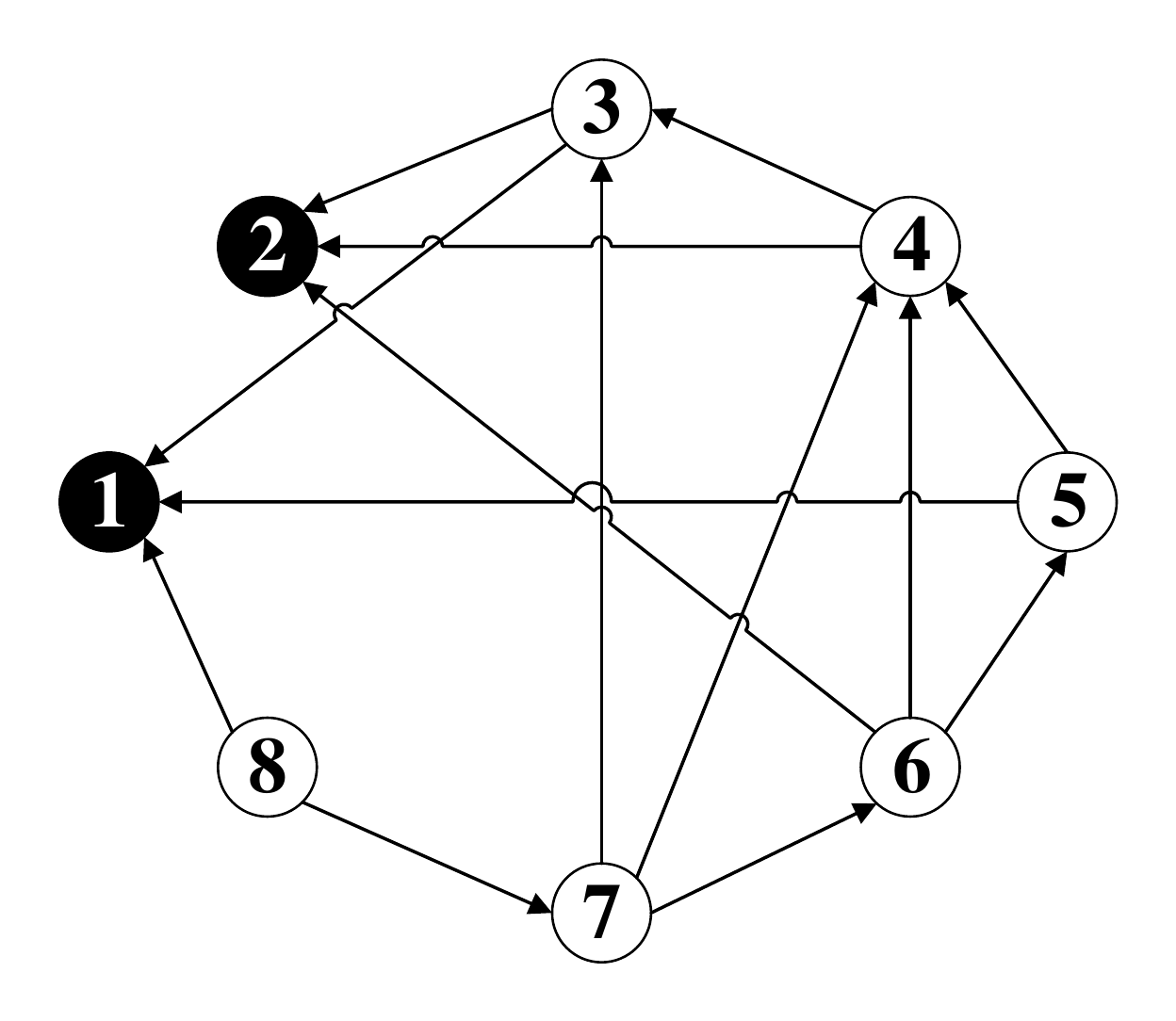}
    \caption{The interaction graph (the black circles represent the leaders)}.
    \label{graph}
\end{figure}
\begin{figure}[h]
    \centering
    \includegraphics[width=0.9\linewidth]{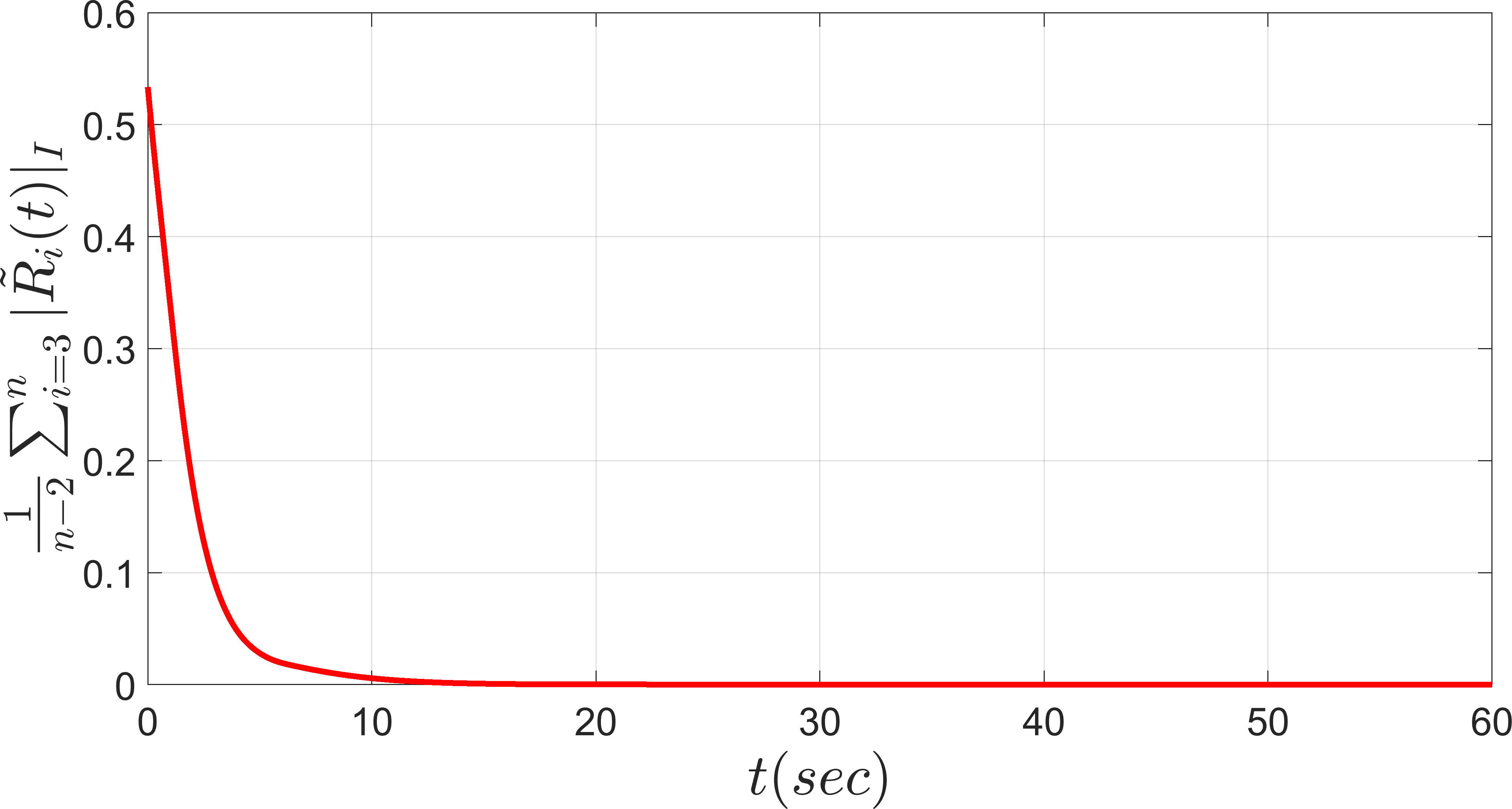}
    \caption{Time evolution of the average attitude estimation error norm.}
    \label{rotation_leader_follower}
\end{figure}
\begin{figure}[h]
    \centering
    \includegraphics[width=0.9\linewidth]{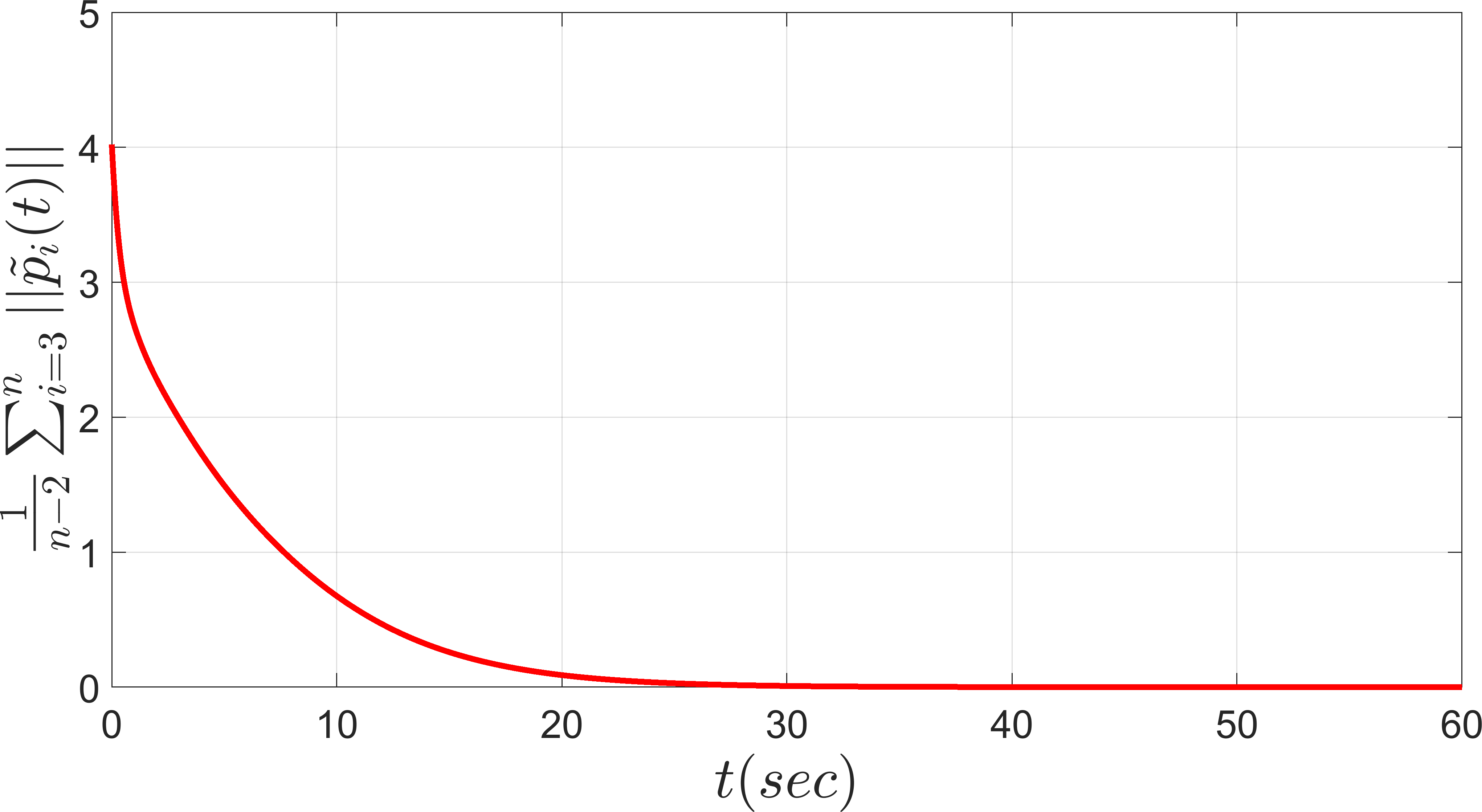}
    \caption{Time evolution of the average position estimation error norm.}
    \label{translation_leader_follower}
\end{figure}

\section{CONCLUSIONS}\label{s6}
In this work, we proposed an almost globally asymptotically stable bearing-based distributed pose estimation scheme for multi-agent networks, with a directed and acyclic interaction graph topology, with two leaders that have access to their respective pose information with respect to a global reference frame. The individual angular velocity and local inter-agent time-varying bearing measurements are assumed to be available. The proposed scheme consists of a cascade of an almost globally asymptotically stable distributed attitude observer and an ISS distributed position observer. It is worth pointing out that the distributed attitude observer is an interesting contribution on its own right as it is stand-alone (\textit{i.e.,} does not depend on the position estimation) and could be used in other applications that involve rotating multi-agent systems. 
An interesting extension of this work, would be the consideration of the case where the agents are allowed to have simultaneous translational and rotational motion.

\addtolength{\textheight}{-12cm}   % This command serves to balance the column lengths

%%%%%%%%%%%%%%%%%%%%%%%%%%%%%%%%%%%%%%%%%%%%%%%%%%%%%%%%%%%%%%%%%%%%%%%%%%%%%%%%

%%%%%%%%%%%%%%%%%%%%%%%%%%%%%%%%%%%%%%%%%%%%%%%%%%%%%%%%%%%%%%%%%%%%%%%%%%%%%%%%

%%%%%%%%%%%%%%%%%%%%%%%%%%%%%%%%%%%%%%%%%%%%%%%%%%%%%%%%%%%%%%%%%%%%%%%%%%%%%%%%

\bibliographystyle{IEEEtran}
\bibliography{References}

\end{document}